\documentclass[12pt]{article}

\usepackage{amsmath,amssymb,amsfonts}
\usepackage{algorithmic}
\usepackage{graphicx}
\usepackage{textcomp}
\usepackage{algorithm}
\usepackage{subcaption}
\usepackage{bm}
\usepackage{amsthm}
\usepackage{float}
\usepackage{placeins}
\usepackage{geometry}
\usepackage{hyperref}
\usepackage{adjustbox}

\newtheorem{theorem}{Theorem}
\newtheorem{lemma}{Lemma}
\newtheorem{corollary}{Corollary}
\newtheorem{proposition}{Proposition}

\title{Re-Rooting-Based Fault-Tolerant One-to-All Broadcasting in Dense Eisenstein--Jacobi Networks}

\author{Bader Albader\\
\small Department of Computer Science, Kuwait University, Kuwait\\
\small \texttt{albader@cs.ku.edu.kw}}

\date{}

\begin{document}

\maketitle

\begin{abstract}
Dense Eisenstein--Jacobi networks are degree-six algebraic interconnection topologies with regular structure, vertex symmetry, small diameter, and efficient communication algorithms. These properties make them suitable for parallel and on-chip communication systems in which collective operations such as one-to-all broadcasting are frequent. Existing optimal broadcasting algorithms for dense hexagonal/Eisenstein--Jacobi networks assume fault-free operation. However, a faulty internal forwarding node may interrupt message propagation and prevent complete delivery. This paper proposes a lightweight re-rooting-based fault-tolerant broadcasting method for dense Eisenstein--Jacobi networks. The main idea is to relocate the effective broadcast source to a new source node such that each faulty node is located at graph distance equal to the network diameter from the new source. Consequently, faulty nodes become leaf-level nodes in the broadcast process and are not required to forward the message. We present source-selection algorithms for one- and two-node failures and prove that for any pair of faulty nodes in a dense Eisenstein--Jacobi network there exists a common distance-diameter node that can serve as a valid re-rooted source. The source-selection procedure requires linear time in the network diameter. Equivalently, since $N=3t^2+3t+1$, the selection cost is $O(\sqrt{N})$ in the number of nodes. Since the standard one-to-all broadcast completes in one diameter time and the relocation phase is also bounded by one diameter, the proposed method completes in at most twice the network diameter. We also show that the two-fault guarantee does not generally extend to arbitrary three-fault configurations by giving an explicit counterexample. The proposed approach improves broadcast reliability without constructing redundant spanning trees, backup paths, or additional broadcast structures.
\end{abstract}


\section{Introduction}
\label{sec:introduction}

Interconnection networks play a central role in the design of parallel computers, network-on-chip architectures, multicore systems, distributed accelerators, and embedded communication fabrics. In such systems, processors or routers must exchange control information, synchronization messages, configuration data, and application data efficiently. Among the most important collective communication primitives is one-to-all broadcasting, in which a source node disseminates a message to every other node in the network.

Two-dimensional toroidal and mesh-based networks have been widely studied because of their simple layout and efficient local communication. However, alternative algebraic networks can provide better degree-diameter trade-offs and more symmetric communication patterns. Dense Gaussian networks and Eisenstein--Jacobi (EJ) networks are two such families. Gaussian networks provide degree-four topologies based on Gaussian integers, whereas EJ networks provide degree-six topologies based on Eisenstein--Jacobi integers. The degree-six EJ structure is closely related to wrap-around hexagonal meshes and admits efficient routing and broadcasting algorithms.

Dense EJ networks are especially attractive because they combine algebraic regularity with a natural hexagonal geometry. The network is vertex-transitive, every node has six neighbors, and the nodes at each distance form a regular hexagonal boundary. For the dense hexagonal network with parameter $n$, the number of nodes is
\begin{equation}
N=3n^2-3n+1,
\end{equation}
and the diameter is
\begin{equation}
t=n-1.
\end{equation}
The standard one-to-all broadcast in this network completes in exactly $t$ parallel steps by propagating the message through six directional sectors.

Although the existing broadcast algorithm is optimal in the fault-free case, it is vulnerable to node failures. If a faulty node appears on an internal forwarding position of the broadcast tree, then all nodes depending on that forwarding node may fail to receive the message. This motivates a fault-tolerant approach that preserves the simplicity of the original broadcast while avoiding the need for redundant spanning trees or adaptive rerouting.

This paper proposes a re-rooting-based fault-tolerant broadcasting method for dense EJ networks. Rather than modifying the broadcast tree or maintaining multiple backup structures, the method dynamically relocates the effective broadcast source. The new source node is selected so that each faulty node lies at graph distance $t$, the network diameter, from the new source. Since nodes at distance $t$ are leaf-level nodes in the standard broadcast process, faulty nodes placed on this boundary are not required to forward packets.

The contribution of this work is specific to the dense Eisenstein--Jacobi topology. Dense EJ networks have a degree-six neighborhood structure, a hexagonal graph-distance geometry, and a six-sector broadcast tree. The proposed method exploits these EJ-specific properties by using the graph-distance-$t$ hexagonal boundary as a relocation target for faulty nodes. The main theoretical argument is based on a boundary-difference coverage property for the EJ network, namely that every node displacement can be represented as the difference of two distance-$t$ boundary nodes. This allows the two-fault source-selection problem to be solved directly within the EJ geometry without modifying the underlying broadcast algorithm.

The main contributions of this paper are as follows:
\begin{itemize}
\item We propose a lightweight source-relocation method for fault-tolerant one-to-all broadcasting in dense Eisenstein--Jacobi networks.
\item We develop source-selection algorithms for one-node and two-node failure scenarios.
\item We prove that for any two faulty nodes in a dense EJ network there exists a node whose graph distance from both faulty nodes is equal to the network diameter.
\item We show that this guarantee does not generally extend to arbitrary three-node failures by presenting an explicit counterexample.
\item We analyze the communication and computational complexity of the method and show that the total broadcast time is bounded by $2t$, where $t$ is the network diameter.
\item We position the proposed method relative to common fault-tolerant broadcasting approaches and discuss its limitations and practical implications.
\end{itemize}

The remainder of this paper is organized as follows. Section~\ref{sec:related} reviews related work. Section~\ref{sec:background} introduces dense EJ networks and their distance structure. Section~\ref{sec:ej_broadcasting} reviews the standard one-to-all broadcasting algorithm. Section~\ref{sec:method} presents the proposed re-rooting method. Section~\ref{sec:proof} proves the two-fault existence result and gives a three-fault counterexample. Section~\ref{sec:complexity} analyzes the cost of the method. Section~\ref{sec:evaluation} presents the experimental evaluation. Section~\ref{sec:conclusion} concludes the paper.

\section{Related Work}
\label{sec:related}

Interconnection-network topology plays an important role in the design of scalable parallel and distributed systems. Regular topologies such as meshes, tori, circulant graphs, Gaussian networks, and Eisenstein--Jacobi networks have been widely studied because their algebraic structure supports efficient routing, broadcasting, and collective communication. Among these topologies, Eisenstein--Jacobi networks are particularly attractive because they provide a degree-six hexagonal structure with small diameter, high symmetry, and efficient communication behavior.

Hexagonal mesh networks were studied in early work on parallel architectures and reliable communication. Chen, Shin, and Kandlur investigated addressing, routing, and broadcasting in hexagonal mesh multiprocessors, while Kandlur and Shin developed reliable broadcast algorithms for HARTS, a hexagonal mesh multicomputer. These works showed that the hexagonal topology provides a natural six-directional communication structure suitable for efficient broadcast and routing. Later work showed that these hexagonal mesh networks can be interpreted algebraically as a special class of Eisenstein--Jacobi networks. In particular, the network \(H_n\) is generated by \(\alpha=n+(n-1)\omega\), has \(3n^2-3n+1\) nodes, and has diameter \(n-1\). This algebraic formulation provides a compact two-coordinate addressing scheme and simplifies shortest-path routing and collective communication.

Eisenstein--Jacobi networks were further studied as algebraic interconnection networks based on the ring of EJ integers. In such networks, nodes are represented as congruence classes modulo an EJ integer, and adjacency is defined by the six unit directions \(\pm1\), \(\pm\omega\), and \(\pm\omega^2\). This yields a regular degree-six topology. The dense family generated by \(\alpha=n+(n-1)\omega\) is especially important because it achieves the maximum number of nodes for the given degree and diameter.

Collective communication algorithms on hexagonal and EJ networks have also been studied extensively. The basic one-to-all broadcast operation sends one message from a source node to all other nodes in the network. Fault tolerance is a central concern in interconnection networks because node or link failures can interrupt message propagation and reduce communication reliability. Classical fault-tolerant broadcasting methods often rely on redundant communication structures, such as multiple spanning trees, independent spanning trees, edge-independent spanning trees, backup paths, or adaptive rerouting.

Fault-tolerant routing has also been widely investigated in Network-on-Chip systems and many-core architectures. These systems often require predictable latency, low routing complexity, and robustness against defective components. Existing methods include adaptive routing around faulty regions, priority-based quality-of-service routing, protection routing, machine-learning-assisted routing, and reinforcement-learning-based mapping.

The method proposed in this paper follows a different approach. Instead of constructing redundant broadcast trees or continuously adapting the route during propagation, the proposed method changes the effective source of the broadcast. The new source \(NS\) is selected so that the faulty nodes lie on the graph-distance-\(t\) boundary of \(NS\), where \(t=n-1\) is the diameter of the dense EJ network. Under the standard one-to-all broadcast tree, boundary nodes are leaf-level nodes and are not required to forward the message. Therefore, once the faulty nodes are placed on this boundary, the original broadcast algorithm can be reused without modification.

This re-rooting approach is topology-specific. It exploits the vertex transitivity, six-directional symmetry, and boundary structure of dense EJ networks. Unlike redundant spanning-tree methods, the proposed method does not require constructing or storing multiple delivery structures. Unlike multiple-path routing, it does not duplicate the broadcast message along several paths. Unlike adaptive routing, it does not require runtime fault-status propagation or local route recomputation during the broadcast.

\section{Dense Eisenstein--Jacobi Networks}
\label{sec:background}

This section introduces the dense EJ network model used throughout the paper.

\subsection{Eisenstein--Jacobi Integers}
Let
\begin{equation}
\omega=\frac{1+i\sqrt{3}}{2}.
\end{equation}
The Eisenstein--Jacobi integers are numbers of the form
\begin{equation}
\mathbb{Z}[\omega]=\{x+y\omega \mid x,y\in\mathbb{Z}\}.
\end{equation}
The six natural directions in the triangular lattice are represented by the units
\begin{equation}
\pm 1,\quad \pm \omega,\quad \pm \omega^2.
\end{equation}
A node can therefore be represented by the two integer coordinates $(x,y)$ corresponding to the EJ number $x+y\omega$.

\subsection{Dense Hexagonal/EJ Network}
Let
\begin{equation}
\alpha=n+(n-1)\omega,
\end{equation}
where $n\geq 2$. The dense EJ network generated by $\alpha$ is the graph
\begin{equation}
H_n=(V,E),
\end{equation}
where the node set is the set of residue classes modulo $\alpha$,
\begin{equation}
V=\mathbb{Z}[\omega]/(\alpha),
\end{equation}
and two nodes $u,v\in V$ are adjacent if
\begin{equation}
v-u\equiv \pm 1,\ \pm\omega,\ \pm\omega^2 \pmod{\alpha}.
\end{equation}
Thus, each node has degree six.

The network has
\begin{equation}
N=3n^2-3n+1
\label{eq:Nn}
\end{equation}
nodes and diameter
\begin{equation}
t=n-1.
\label{eq:diameter}
\end{equation}
Equivalently, in terms of the diameter $t$, the number of nodes is
\begin{equation}
N=3t^2+3t+1.
\label{eq:Nt}
\end{equation}

The dense EJ network is vertex-transitive. Therefore, distance-based arguments can be carried out with respect to the origin without loss of generality and then translated to any other node.

\subsection{Distance in the Hexagonal Lattice}
For a node represented by $x+y\omega$, the graph distance from the origin in the infinite triangular lattice is
\begin{equation}
D(x,y)=
\begin{cases}
|x+y|, & xy\geq 0,\\
\max\{|x|,|y|\}, & xy<0.
\end{cases}
\label{eq:hex_distance}
\end{equation}
Equivalently,
\begin{equation}
D(x,y)=\max\{|x|,|y|,|x+y|\}.
\label{eq:hex_distance_max}
\end{equation}
The nodes at distance $j$ from any fixed node form a hexagonal boundary of size $6j$, for $1\leq j\leq t$. In particular, the graph-distance-$t$ boundary contains $6t$ nodes.

We define the boundary of the origin by
\begin{equation}
B_t=\{z\in H_n\mid d(0,z)=t\}.
\label{eq:boundary}
\end{equation}
By vertex transitivity, the boundary of a node $u$ is
\begin{equation}
B_t(u)=u+B_t.
\label{eq:boundary_translate}
\end{equation}

\subsection{Circulant Representation}
The dense EJ network $H_n$ is isomorphic to a degree-six circulant graph with
\begin{equation}
N=3n^2-3n+1
\end{equation}
nodes and jumps
\begin{equation}
n-1,\quad n,\quad 2n-1.
\end{equation}
In terms of $t=n-1$, the jumps are
\begin{equation}
t,\quad t+1,\quad 2t+1.
\end{equation}
This representation is useful for implementation and simulation, while the EJ representation is useful for geometric reasoning and proof construction.

\subsection{Integer Labeling of EJ Nodes}
\label{subsec:integer_labeling}

Although the EJ coordinate representation \(x+y\omega\) is geometrically convenient, an integer-label representation is useful for simulations, examples, and comparison with circulant-network descriptions. In the dense EJ network generated by
\begin{equation}
	\alpha=n+(n-1)\omega,
\end{equation}
the total number of nodes is
\begin{equation}
	N=3n^2-3n+1.
\end{equation}
The corresponding diameter is
\begin{equation}
	t=n-1.
\end{equation}

We define the integer-labeling function
\begin{equation}
	\phi:H_n\rightarrow \mathbb{Z}_N
\end{equation}
by
\begin{equation}
	\phi(x+y\omega)\equiv (n-1)x+(2n-1)y \pmod N.
	\label{eq:ej_integer_label}
\end{equation}
Thus, each EJ coordinate node \(x+y\omega\) is assigned a unique integer label modulo \(N\).

This labeling is consistent with the degree-six circulant representation of the dense EJ network. Indeed,
\begin{equation}
	\phi(1)=n-1,
\end{equation}
and
\begin{equation}
	\phi(\omega)=2n-1.
\end{equation}
Since
\begin{equation}
	\omega^2=\omega-1,
\end{equation}
we also have
\begin{equation}
	\phi(\omega^2)=\phi(\omega-1)=(2n-1)-(n-1)=n.
\end{equation}
Therefore, the six EJ adjacency directions
\begin{equation}
	\pm1,\ \pm\omega,\ \pm\omega^2
\end{equation}
correspond in the integer representation to the six circulant jumps
\begin{equation}
	\pm(n-1),\ \pm(2n-1),\ \pm n
\end{equation}
modulo \(N\). Hence, the EJ network \(H_n\) is equivalently represented as the degree-six circulant graph
\begin{equation}
	C_N(n-1,n,2n-1).
\end{equation}

The labeling function also preserves relative node differences. For any two EJ nodes \(U,V\in H_n\),
\begin{equation}
	\phi(U-V)\equiv \phi(U)-\phi(V)\pmod N.
	\label{eq:ej_difference_preserving}
\end{equation}
This property allows relative positions, wrap-around adjacency, and source-selection conditions to be expressed using modular integer arithmetic.

\begin{figure}[H]
	\centering
	\begin{subfigure}{0.48\linewidth}
		\centering
		\includegraphics[width=\linewidth]{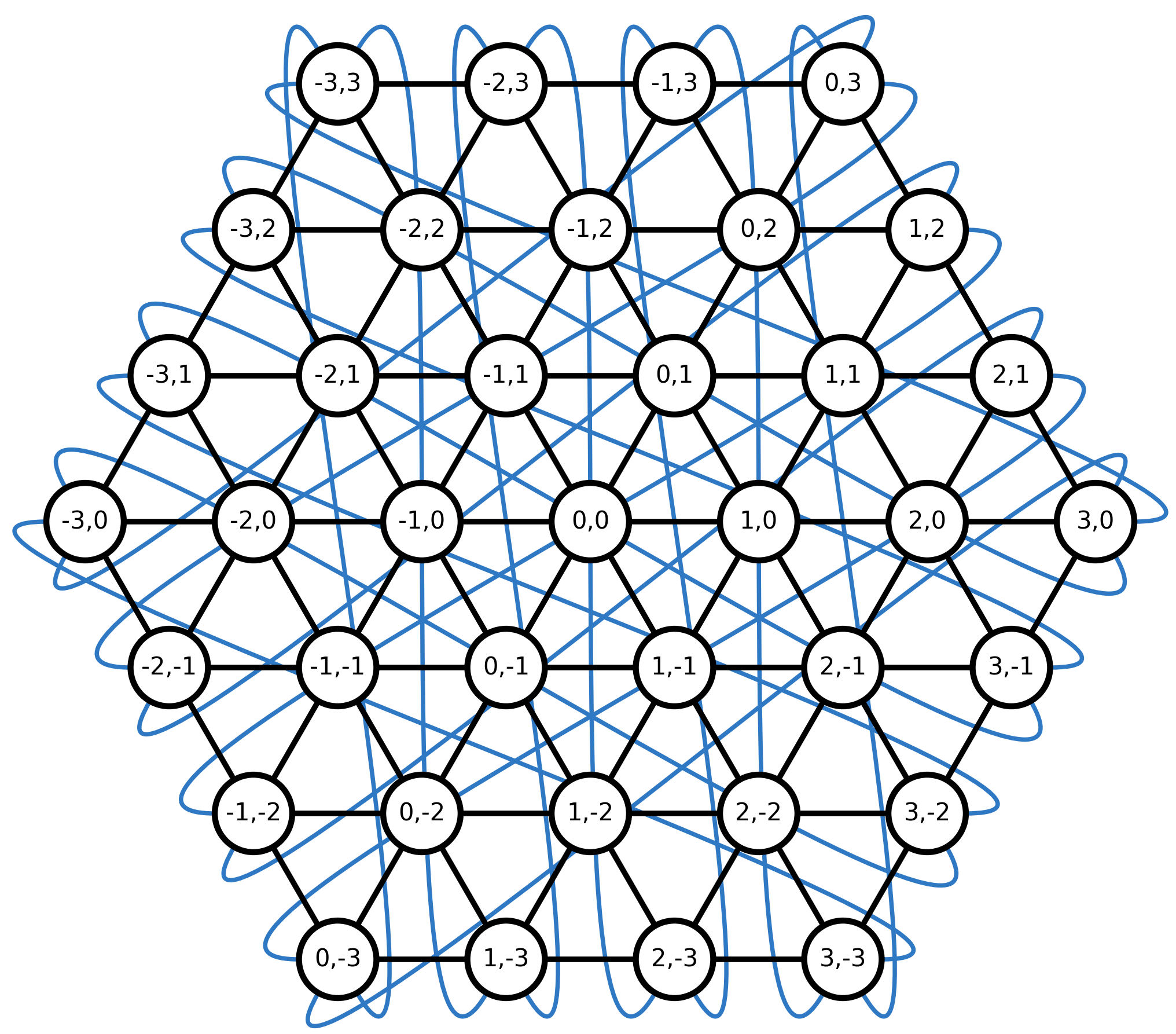}
		\caption{EJ coordinate representation}
	\end{subfigure}
	\hfill
	\begin{subfigure}{0.48\linewidth}
		\centering
		\includegraphics[width=\linewidth]{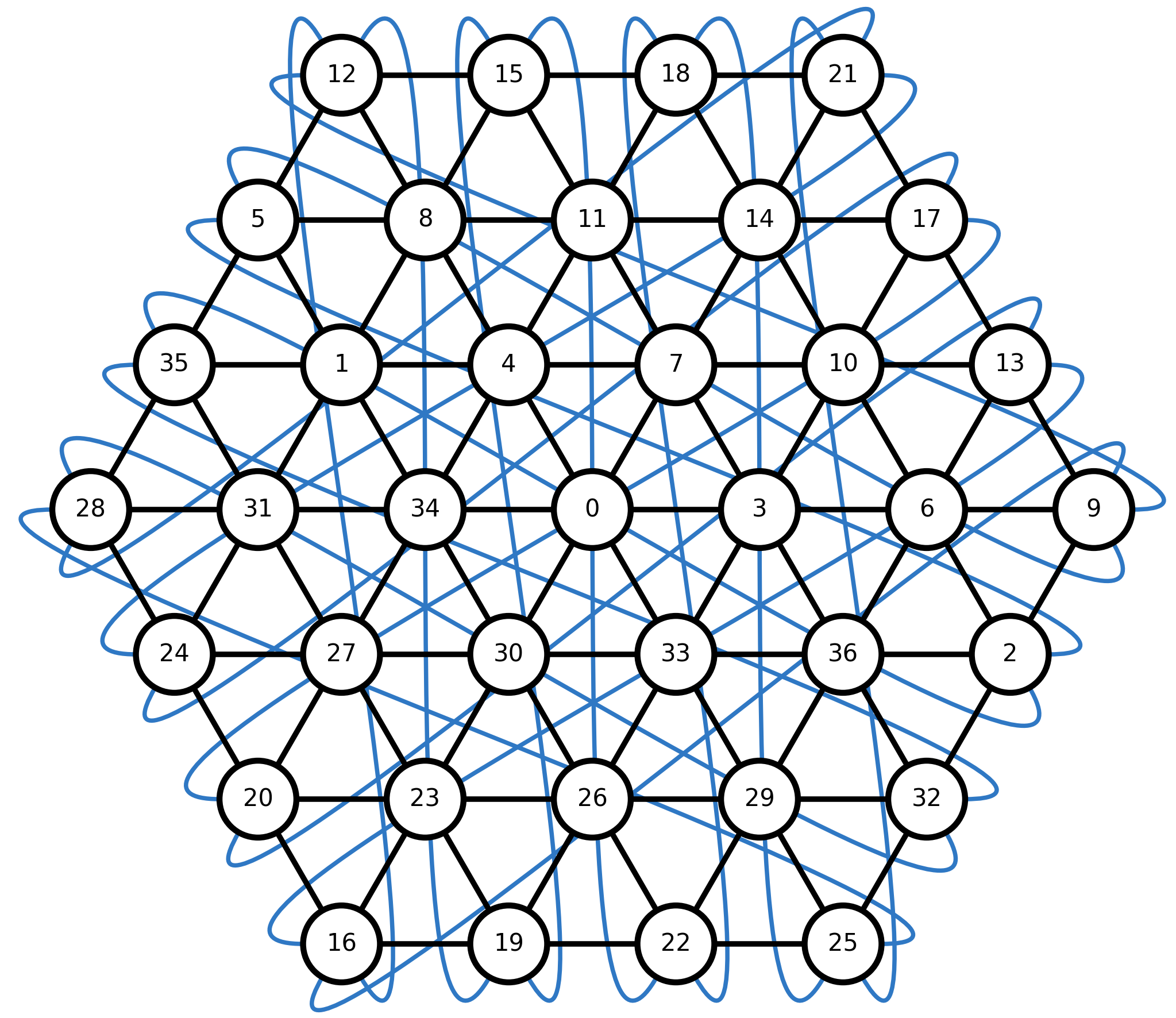}
		\caption{Integer-label representation}
	\end{subfigure}
	\caption{Dense Eisenstein--Jacobi network \(H_4\) generated by \(\alpha=4+3\omega\). The network has diameter \(t=3\) and \(N=37\) nodes. The left figure shows the EJ coordinate representation \(x+y\omega\), while the right figure shows the corresponding integer-label representation obtained from \(\phi(x+y\omega)\equiv 3x+7y \pmod{37}\). The blue curves represent wrap-around links induced by the modular EJ topology.}
	\label{fig:ej_h4_coordinates_integer}
\end{figure}

For example, consider \(H_4\). In this case,
\begin{equation}
	n=4,\qquad t=3,\qquad N=3(4)^2-3(4)+1=37.
\end{equation}
Using \eqref{eq:ej_integer_label}, the six neighbors of node \(0\) are labeled as follows:
\begin{align}
	\phi(1) &\equiv 3 \pmod{37},\\
	\phi(\omega) &\equiv 7 \pmod{37},\\
	\phi(\omega^2)=\phi(-1+\omega) &\equiv 4 \pmod{37},\\
	\phi(-1) &\equiv 34 \pmod{37},\\
	\phi(-\omega) &\equiv 30 \pmod{37},\\
	\phi(-\omega^2)=\phi(1-\omega) &\equiv 33 \pmod{37}.
\end{align}
Thus, node \(0\) is adjacent to nodes
\begin{equation}
	3,\ 4,\ 7,\ 30,\ 33,\ 34
\end{equation}
in the integer-label representation, which are exactly the jumps
\begin{equation}
	\pm3,\ \pm4,\ \pm7
\end{equation}
modulo \(37\).

Figure~\ref{fig:ej_h4_coordinates_integer} illustrates the same \(H_4\) network using both EJ coordinates and integer labels. The left representation is useful for geometric reasoning, while the right representation is useful for simulations and circulant-graph implementation.

\section{One-to-All Broadcasting in Dense EJ Networks}
\label{sec:ej_broadcasting}

In this section, we review the standard one-to-all broadcasting procedure in dense Eisenstein--Jacobi networks.

Let \(H_n\) be the dense EJ network generated by
\begin{equation}
	\alpha=n+(n-1)\omega.
\end{equation}
The network diameter is
\begin{equation}
	t=n-1.
\end{equation}
Since \(H_n\) is homogeneous and vertex-transitive, the one-to-all broadcast can be described without loss of generality from source node \(0\).

The EJ network has six natural communication directions:
\begin{equation}
	1,\ \omega,\ \omega^2,\ -1,\ -\omega,\ -\omega^2.
\end{equation}
These correspond to the six geometric directions
\begin{equation}
	E,\ NE,\ NW,\ W,\ SW,\ SE,
\end{equation}
respectively. Thus, the broadcast expands from the source through six sectors of the hexagonal topology.

At the beginning of the broadcast, the source sends the message to its six neighbors. Each neighbor becomes the root of one sector of the broadcast tree. The broadcast then proceeds outward in parallel through the six sectors. In each sector, the forwarding pattern follows the shortest-path structure of the hexagonal lattice. Nodes on the main sector axis forward the packet to two outward directions, while internal sector nodes forward the packet along one direction. This prevents duplicate packet reception and ensures that every node receives the message exactly once.

\begin{figure}[H]
	\centering
	\includegraphics[width=0.55\linewidth]{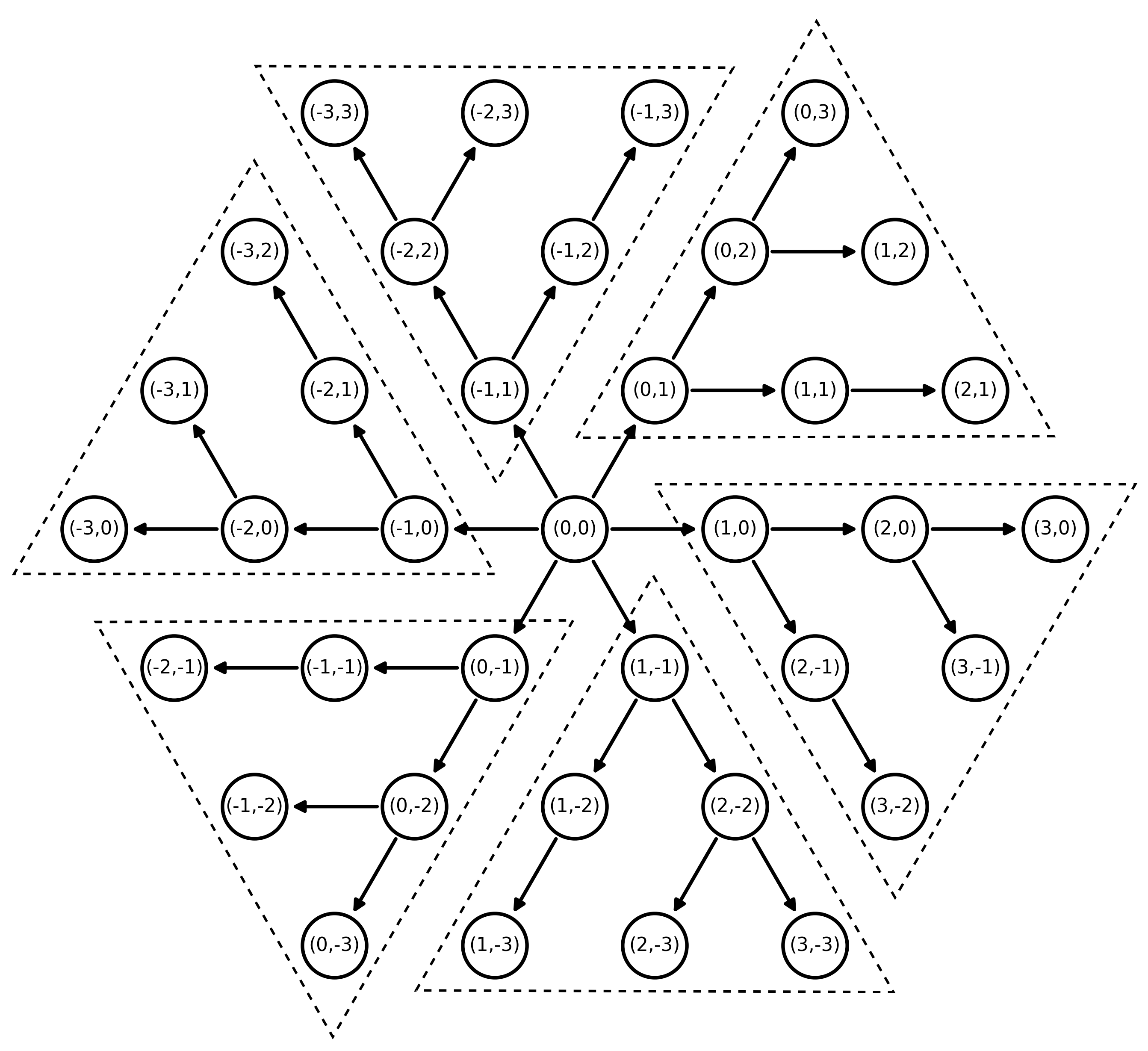}
	\caption{One-to-all broadcasting tree in the dense Eisenstein--Jacobi network \(H_4\), where the source node is \(0\). The broadcast expands through the six EJ directions and covers the six sectors of the hexagonal topology. Boundary nodes at graph distance \(t=3\) are leaf-level nodes and do not forward the message further.}
	\label{fig:ej_broadcast_tree}
\end{figure}

Figure~\ref{fig:ej_broadcast_tree} illustrates the six-sector structure of the broadcast process. The important property for the proposed fault-tolerant method is that nodes at graph distance \(t\) from the source are leaf-level nodes. Such nodes may receive the message, but they are not required to forward it to any other node. Therefore, a faulty node located at distance \(t\) from the source does not disrupt the broadcast process.

A simple packet-header structure can be used to implement the broadcast. The packet contains a broadcast-control field, a distance field, and a six-bit direction field. The distance field is initialized to the network diameter \(t\). The six-bit direction field, denoted by \(DIR\), indicates which of the six output ports are enabled for forwarding. In this paper, the bits of \(DIR\) are ordered as
\begin{equation}
	(E,\ SE,\ SW,\ W,\ NW,\ NE).
	\label{eq:dir_bit_order}
\end{equation}
Thus, the one-direction masks are interpreted as
\begin{equation}
	100000=E,\quad
	010000=SE,\quad
	001000=SW,
\end{equation}
\begin{equation}
	000100=W,\quad
	000010=NW,\quad
	000001=NE.
\end{equation}
At each hop, the distance field is decremented. When the distance field reaches zero, the packet is consumed but not forwarded. Thus, the broadcast terminates exactly at the graph-distance-\(t\) boundary.

The forwarding rule is summarized in Algorithm~\ref{alg:ej_broadcast_forwarding}, and the corresponding direction-mask actions are listed in Table~\ref{tab:direction_masks}.

\begin{algorithm}[H]
	\caption{One-to-all broadcast forwarding in \(H_n\)}
	\label{alg:ej_broadcast_forwarding}
	\begin{algorithmic}[1]
		\REQUIRE Distance field \(dist\), direction field \(DIR\)
		\STATE The six-bit direction field is ordered as \((E,SE,SW,W,NW,NE)\).
		\IF{\(dist=t\)}
		\STATE Send to \(E,NE,NW,W,SW,SE\) using initial sector masks
		\STATE \(110000,100001,000011,000110,001100,011000\), respectively
		\ELSIF{\(dist=0\)}
		\STATE Consume packet and stop forwarding
		\ELSE
		\STATE Consume packet and set \(dist\leftarrow dist-1\)
		\STATE Apply the sector forwarding rules in Table~\ref{tab:direction_masks}
		\ENDIF
	\end{algorithmic}
\end{algorithm}

\begin{table}[H]
	\caption{Direction-mask forwarding rules used by Algorithm~\ref{alg:ej_broadcast_forwarding}.}
	\label{tab:direction_masks}
	\centering
	\begin{tabular}{c c}
		\hline
		Incoming mask & Forwarding action \\
		\hline
		110000 & \(E:110000,\ SE:010000\) \\
		100001 & \(NE:100001,\ E:100000\) \\
		000011 & \(NW:000011,\ NE:000001\) \\
		000110 & \(W:000110,\ NW:000010\) \\
		001100 & \(SW:001100,\ W:000100\) \\
		011000 & \(SE:011000,\ SW:001000\) \\
		100000 & \(E:100000\) \\
		010000 & \(SE:010000\) \\
		001000 & \(SW:001000\) \\
		000100 & \(W:000100\) \\
		000010 & \(NW:000010\) \\
		000001 & \(NE:000001\) \\
		\hline
	\end{tabular}
\end{table}

Table~\ref{tab:direction_masks} lists the forwarding actions used by Algorithm~\ref{alg:ej_broadcast_forwarding}. In each two-direction forwarding case, the first outgoing copy preserves the current sector-axis mask, while the second outgoing copy uses a one-direction mask.

\begin{figure}[H]
	\centering
	\begin{subfigure}{0.48\linewidth}
		\centering
		\textbf{(A)}\\[2pt]
		\includegraphics[width=\linewidth]{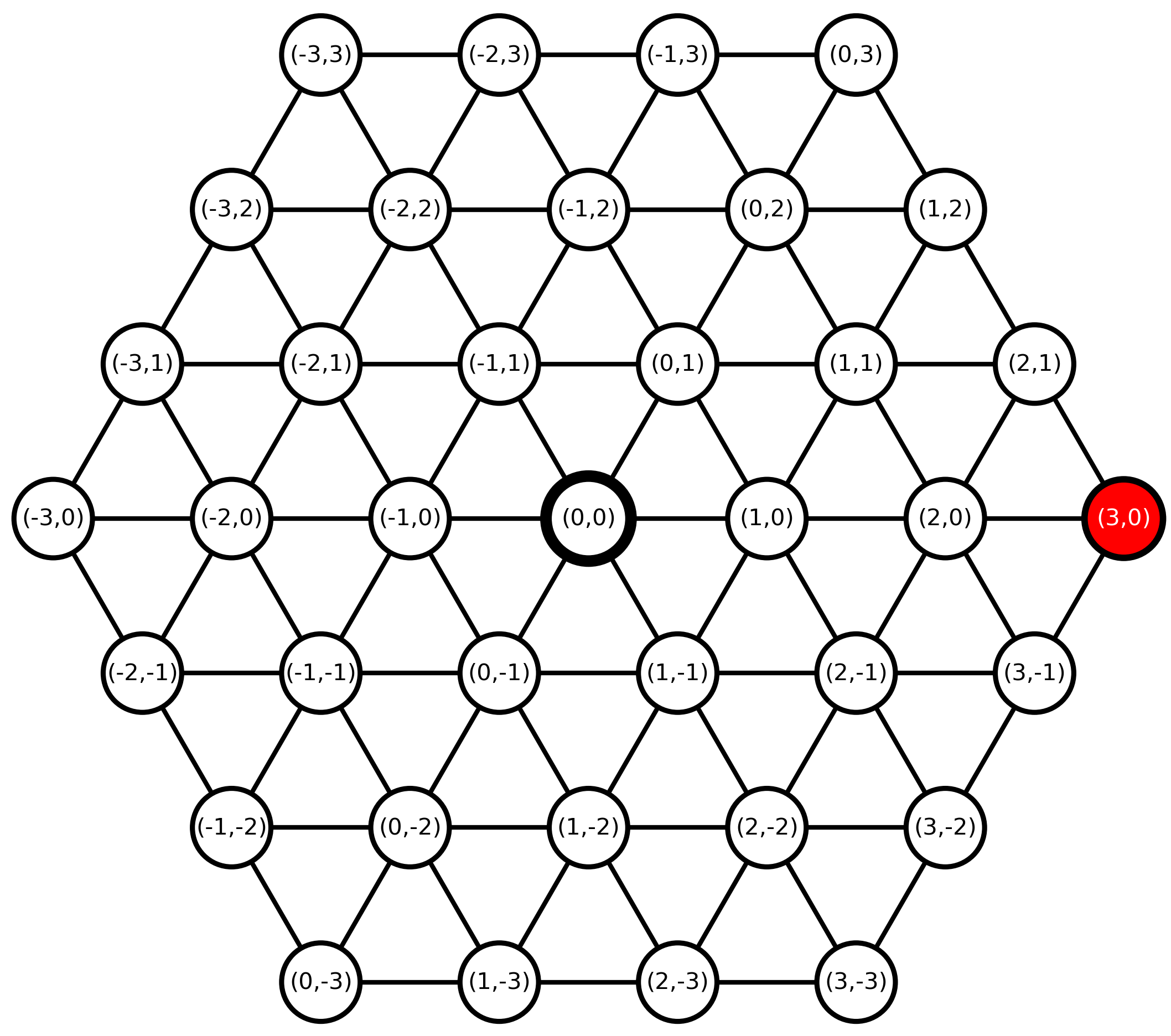}
	\end{subfigure}
	\hfill
	\begin{subfigure}{0.48\linewidth}
		\centering
		\textbf{(B)}\\[2pt]
		\includegraphics[width=\linewidth]{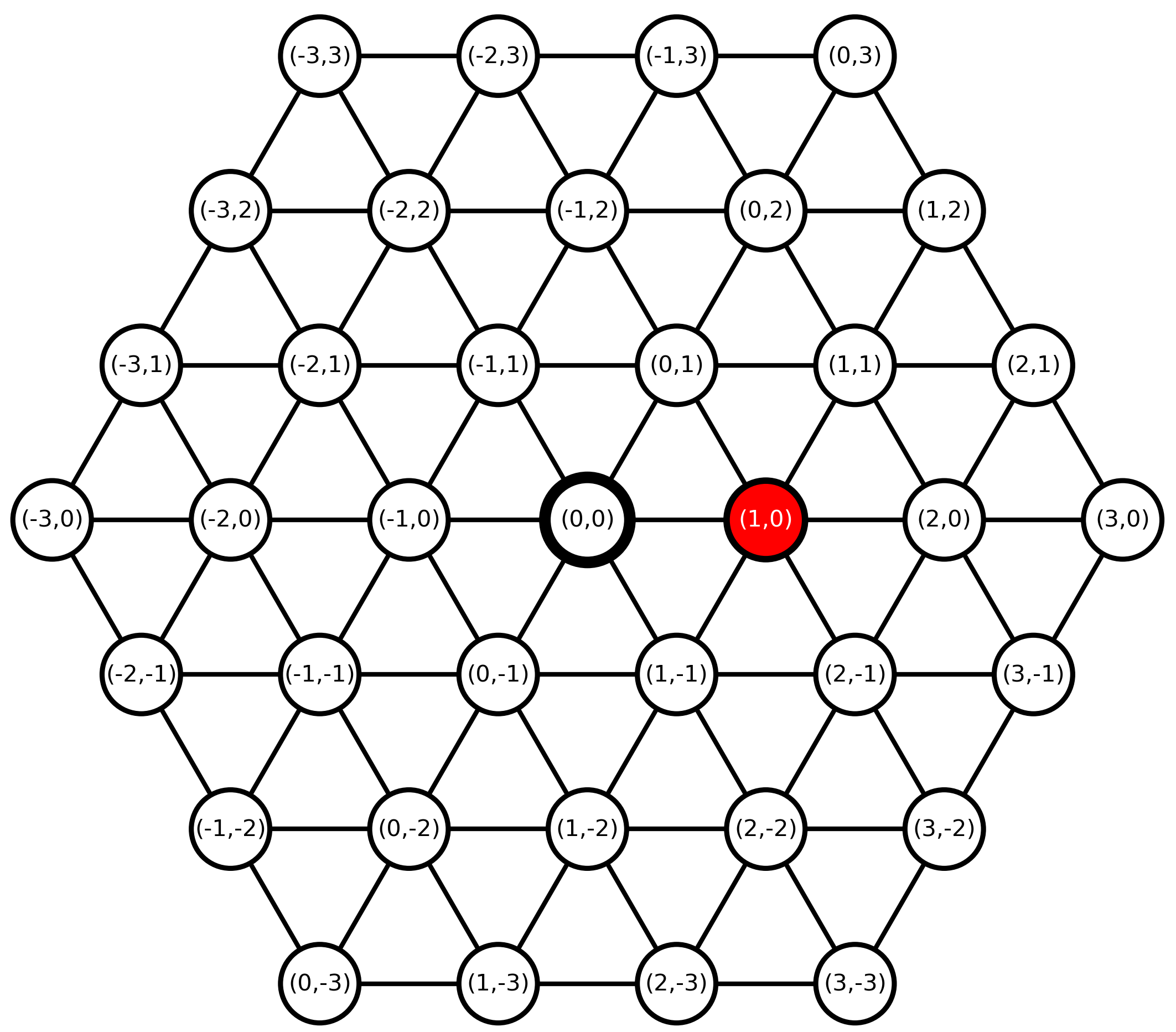}
	\end{subfigure}
	\caption{Fault classification in the dense Eisenstein--Jacobi network \(H_4\) with diameter \(t=3\), rooted at the source node \((0,0)\). (A) A faulty node located on the graph-distance-\(t\) boundary. (B) A faulty node located at graph distance less than \(t\).}
	\label{fig:ej_fault_classification}
\end{figure}

The algorithm is independent of the identity of the source node. If the source is \(S\) instead of \(0\), the same forwarding structure is obtained by translating all nodes by \(S\).

The broadcast completes in exactly \(t=n-1\) parallel steps. At step \(j\), for \(1\leq j\leq t\), the broadcast reaches the nodes at graph distance \(j\) from the source. Since the dense EJ network has \(6j\) nodes at graph distance \(j\), the number of newly reached nodes at each step grows linearly with \(j\). At the final step, the broadcast reaches the boundary nodes at distance \(t\). These boundary nodes are leaves of the broadcast tree.

The total number of non-source nodes is
\begin{equation}
	N-1=3n^2-3n=3t^2+3t.
\end{equation}
Also,
\begin{equation}
	\sum_{j=1}^{t}6j=3t(t+1)=3t^2+3t=N-1.
\end{equation}
Therefore, the broadcast reaches every non-source node exactly once. Since no duplicate messages are generated, the broadcast uses exactly \(N-1\) tree edges. This is optimal for one-to-all broadcasting because any broadcast tree spanning all \(N\) nodes must use at least \(N-1\) edges.

The communication time of the fault-free one-to-all broadcast is therefore
\begin{equation}
	B_{\mathrm{baseline}}(t)=t.
\end{equation}
Equivalently, in terms of \(n\),
\begin{equation}
	B_{\mathrm{baseline}}(n)=n-1.
\end{equation}

This baseline broadcast algorithm is efficient but not fault tolerant. If a faulty node occurs on the distance-\(t\) boundary of the source, the broadcast remains correct because the faulty node is a leaf-level node and does not need to forward the message. However, if a faulty node occurs at graph distance less than \(t\), it may be an internal forwarding node. In that case, packets that should pass through that node may not reach all downstream nodes in the corresponding sector. The proposed method addresses this problem by relocating the effective source before starting the broadcast so that the faulty nodes become distance-\(t\) boundary nodes with respect to the new source.

\section{Proposed Re-Rooting Method}
\label{sec:method}

We now present the proposed fault-tolerant broadcasting method. The method assumes a static node-failure model: the faulty nodes are known before the broadcast begins and do not change during propagation. The original source node is assumed to be operational.

\subsection{Fault Classification}
\label{subsec:fault_classification}

Let \(S\) be the original broadcast source in the dense EJ network \(H_n\), and let
\begin{equation}
	t=n-1
\end{equation}
be the network diameter. A faulty node \(F\) can be classified according to its graph distance from the source \(S\).

The first case occurs when
\begin{equation}
	d(S,F)=t.
\end{equation}
In this case, the faulty node lies on the graph-distance-\(t\) boundary of the source. Under the standard EJ one-to-all broadcast, boundary nodes are leaf-level nodes. They may receive the broadcast message, but they are not required to forward it to any other node. Therefore, a faulty node at distance \(t\) from the source does not interrupt the propagation of the broadcast tree.

The second case occurs when
\begin{equation}
	d(S,F)<t.
\end{equation}
In this case, the faulty node may lie at an internal forwarding position of the broadcast tree. If such a node fails to forward the message, then all downstream nodes that depend on that forwarding step may not receive the broadcast. Therefore, faults at distance less than \(t\) can disrupt the correctness of the original one-to-all broadcast.

For multiple faulty nodes, the same classification applies to each faulty node. If all faulty nodes are located at distance \(t\) from the source, then the original broadcast can proceed without re-rooting. However, if at least one faulty node is located at distance less than \(t\), then the broadcast source must be relocated to a new source node \(NS\). The goal of re-rooting is to choose \(NS\) so that every faulty node becomes a graph-distance-\(t\) boundary node with respect to \(NS\).

Thus, for a faulty set \(F=\{F_1,F_2,\ldots,F_m\}\), the desired re-rooted source satisfies
\begin{equation}
	d(NS,F_i)=t,\qquad 1\leq i\leq m.
	\label{eq:reroot_condition_general}
\end{equation}
When this condition holds, all faulty nodes are leaf-level nodes in the broadcast tree rooted at \(NS\), and none of them is required to forward the broadcast message.

Figure~\ref{fig:ej_fault_classification} illustrates the two fault categories. The proposed re-rooting method addresses the disruptive case by selecting a new source \(NS\) that moves the faulty nodes to the graph-distance-\(t\) boundary.

\subsection{Main Idea}
Let $S$ be the original source and let $F$ be the set of faulty nodes. If every faulty node is already at graph distance $t$ from $S$, then the standard broadcast can proceed from $S$ because all faulty nodes are leaf-level nodes. Otherwise, the method selects a new source node $NS$ such that
\begin{equation}
d(NS,f)=t,\quad \text{for every } f\in F.
\label{eq:ns_condition}
\end{equation}
The message is then routed from $S$ to $NS$, and the standard one-to-all broadcast is initiated from $NS$.

Since every faulty node is located on the graph-distance-$t$ boundary with respect to $NS$, faulty nodes are reached only at the final broadcast step and are not required to forward the message. Thus, the original broadcast algorithm is preserved.

\subsection{Single-Fault Case}
For a single faulty node $F_1$, a valid new source can be chosen from the boundary
\begin{equation}
B_t(F_1)=F_1+B_t.
\end{equation}
Any node in this set is at graph distance $t$ from $F_1$. If the selected node is faulty or unsuitable for implementation reasons, another boundary node can be chosen.

\begin{algorithm}[H]
	\caption{OneFailureFindingNS}
	\label{alg:one_failure_ns}
	\begin{algorithmic}[1]
		\REQUIRE Faulty node \(F_1\), network parameter \(n\), diameter \(t=n-1\)
		\ENSURE New source node \(NS\)
		\STATE Generate the graph-distance-\(t\) boundary set \(B_t\)
		\FOR{each boundary node \(U\in B_t\)}
		\STATE \(NS\leftarrow F_1+U \pmod{\alpha}\)
		\IF{\(NS\) is operational}
		\STATE \textbf{return} \(NS\)
		\ENDIF
		\ENDFOR
	\end{algorithmic}
\end{algorithm}

Algorithm~\ref{alg:one_failure_ns} gives a simple source-selection method that scans the distance-\(t\) boundary \(B_t\), which contains exactly \(6t\) nodes. Therefore, the worst-case running time is
\begin{equation}
	T_1(t)=O(6t)=O(t).
\end{equation}

\subsection{Two-Fault Case}
For two faulty nodes $F_1$ and $F_2$, the goal is to find a node $NS$ satisfying
\begin{equation}
d(F_1,NS)=d(F_2,NS)=t.
\label{eq:two_fault_condition}
\end{equation}
Equivalently,
\begin{equation}
NS\in B_t(F_1)\cap B_t(F_2).
\label{eq:boundary_intersection}
\end{equation}
The existence of such a node for every pair of faulty nodes is proved in Section~\ref{sec:proof}.

\begin{algorithm}[H]
	\caption{TwoFailureFindingNS}
	\label{alg:two_failure_ns}
	\begin{algorithmic}[1]
		\REQUIRE Faulty nodes \(F_1,F_2\), network parameter \(n\), diameter \(t=n-1\)
		\ENSURE New source node \(NS\)
		\STATE Generate the graph-distance-\(t\) boundary set \(B_t\)
		\STATE \(A\leftarrow F_2-F_1 \pmod{\alpha}\)
		\FOR{each boundary node \(U\in B_t\)}
		\STATE \(V\leftarrow U-A \pmod{\alpha}\)
		\IF{\(V\in B_t\)}
		\STATE \(NS\leftarrow F_1+U \pmod{\alpha}\)
		\IF{\(NS\) is operational}
		\STATE \textbf{return} \(NS\)
		\ENDIF
		\ENDIF
		\ENDFOR
	\end{algorithmic}
\end{algorithm}

Algorithm~\ref{alg:two_failure_ns} uses translation symmetry. First, it reduces the problem to finding a boundary node for the relative difference
\begin{equation}
	A=F_2-F_1.
\end{equation}
Then it translates the solution back to the original position. Since \(|B_t|=6t\), and membership in \(B_t\) can be checked in constant time using the distance formula or a precomputed boundary set, the worst-case running time is
\begin{equation}
	T_2(t)=O(6t)=O(t).
\end{equation}

\subsection{Worked Example}
\label{subsec:worked_example}

We now give a concrete example to illustrate the two-fault re-rooting idea. Consider the dense EJ network \(H_4\), generated by
\begin{equation}
	\alpha=4+3\omega.
\end{equation}
In this case,
\begin{equation}
	n=4,\qquad t=n-1=3,\qquad N=3n^2-3n+1=37.
\end{equation}

Assume that the two faulty nodes are
\begin{equation}
	F_1=(0,0)
\end{equation}
and
\begin{equation}
	F_2=(1,0).
\end{equation}
Using the integer-labeling function
\begin{equation}
	\phi(x+y\omega)\equiv 3x+7y \pmod{37},
\end{equation}
these nodes correspond to
\begin{equation}
	\phi(F_1)=0
\end{equation}
and
\begin{equation}
	\phi(F_2)=3.
\end{equation}

The goal is to find a new source node \(NS\) satisfying
\begin{equation}
	d(NS,F_1)=d(NS,F_2)=3.
\end{equation}
One valid choice is
\begin{equation}
	NS=(-3,0).
\end{equation}
For the first faulty node,
\begin{align}
	d(F_1,NS)
	&=D(-3,0)\\
	&=\max\{|-3|,\ |0|,\ |-3+0|\}\\
	&=3.
\end{align}

For the second faulty node, the displacement from \(F_2\) to \(NS\) is
\begin{equation}
	NS-F_2=(-3,0)-(1,0)=(-4,0).
\end{equation}
Modulo \(\alpha=4+3\omega\), this displacement is equivalent to
\begin{equation}
	(-4,0)\equiv (0,-3)\pmod{4+3\omega}.
\end{equation}
Therefore,
\begin{align}
	d(F_2,NS)
	&=D(0,-3)\\
	&=\max\{|0|,\ |-3|,\ |0-3|\}\\
	&=3.
\end{align}
Hence,
\begin{equation}
	d(F_1,NS)=d(F_2,NS)=3.
\end{equation}

In integer labels, the selected new source is
\begin{equation}
	\phi(-3,0)\equiv 3(-3)+7(0)\equiv -9\equiv 28\pmod{37}.
\end{equation}
Thus, the same example can be written as
\begin{equation}
	F_1=0,\qquad F_2=3,\qquad NS=28.
\end{equation}
The node \(28\) is at graph distance \(3\) from both faulty nodes \(0\) and \(3\).

\begin{figure}[H]
	\centering
	\includegraphics[width=0.6\linewidth]{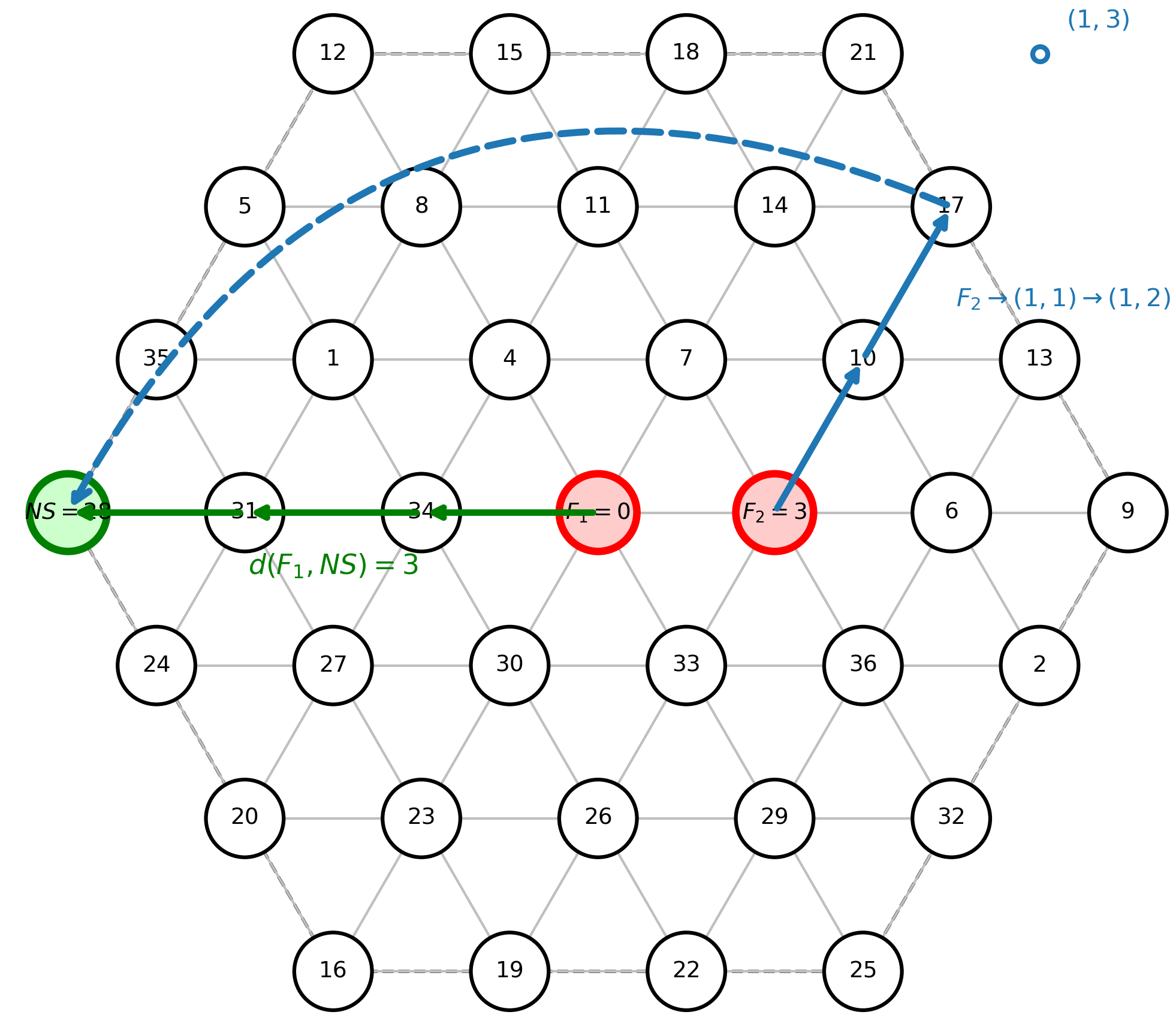}
	\caption{Worked example of two-fault re-rooting in \(H_4\). The faulty nodes are \(F_1=0\) and \(F_2=3\), and one valid new source is \(NS=28\). The path from \(F_2\) to \(NS\) uses the wrap-around relation \((-4,0)\equiv(0,-3)\pmod{4+3\omega}\), showing that \(d(F_2,NS)=3\). Thus, both faulty nodes are at graph distance \(t=3\) from the new source.}
	\label{fig:two_fault_example}
\end{figure}

After \(NS\) is selected, the message is first routed from the original source \(S\) to \(NS\). Then, the standard one-to-all EJ broadcast is executed from \(NS\). Since both faulty nodes are now located on the graph-distance-\(3\) boundary of \(NS\), they are leaf-level nodes and do not participate in forwarding.

\section{Existence Proof and Theoretical Boundary}
\label{sec:proof}

This section proves the main theoretical guarantee for two faulty nodes and then shows why the guarantee cannot be extended to all three-fault configurations.

\subsection{Boundary-Difference Coverage}
\label{subsec:boundary_difference}

The key step in the two-fault re-rooting argument is to show that every possible relative displacement between two faulty nodes can be represented as the difference of two boundary nodes. This property allows us to translate a two-fault problem into a common-boundary problem.

Let
\begin{equation}
	B_t=\{z\in H_n\mid d(0,z)=t\}
\end{equation}
denote the graph-distance-\(t\) boundary around node \(0\), where
\begin{equation}
	t=n-1
\end{equation}
is the diameter of \(H_n\). Since \(H_n\) is vertex-transitive, the graph-distance-\(t\) boundary around any node is obtained by translating \(B_t\).

For an EJ coordinate \(z=x+y\omega\), the distance from the origin in the hexagonal coordinate representation is
\begin{equation}
	D(x,y)=\max\{|x|,\ |y|,\ |x+y|\}.
	\label{eq:ej_hex_distance}
\end{equation}
Thus, the closed hexagonal ball of radius \(t\) is
\begin{equation}
	\mathcal{H}_t=\{x+y\omega\mid D(x,y)\leq t\}.
\end{equation}
The boundary of this ball is
\begin{equation}
	B_t=\{x+y\omega\mid D(x,y)=t\}.
\end{equation}

\begin{lemma}[Boundary-Difference Coverage]
	\label{lem:boundary_difference}
	For the dense Eisenstein--Jacobi network \(H_n\) generated by
	\begin{equation}
		\alpha=n+(n-1)\omega,
	\end{equation}
	with diameter \(t=n-1\), the graph-distance-\(t\) boundary satisfies
	\begin{equation}
		B_t-B_t=H_n.
		\label{eq:boundary_difference_coverage}
	\end{equation}
	That is, for every node \(A\in H_n\), there exist two boundary nodes \(U,V\in B_t\) such that
	\begin{equation}
		A=U-V.
	\end{equation}
\end{lemma}

\begin{proof}
	Let \(A\in H_n\) be arbitrary. Because \(H_n\) is the dense EJ network generated by \(\alpha=n+(n-1)\omega\), every node has a representative inside the closed hexagonal ball \(\mathcal{H}_t\). This representative can be chosen by reducing the EJ coordinate modulo \(\alpha\) to the canonical diameter-\(t\) region; equivalently, since \(H_n\) has diameter \(t\), each residue class contains a coordinate representative at graph distance at most \(t\) from the origin. Therefore, we may write
	\begin{equation}
		A=x+y\omega
	\end{equation}
	with
	\begin{equation}
		D(x,y)\leq t.
	\end{equation}
	It is sufficient to prove that every point in \(\mathcal{H}_t\) can be written as the difference of two points on the boundary \(B_t\).
	
	The hexagonal ball \(\mathcal{H}_t\) is divided into six sectors corresponding to the six EJ directions
	\begin{equation}
		1,\ \omega,\ \omega^2,\ -1,\ -\omega,\ -\omega^2.
	\end{equation}
	We first prove the construction explicitly for the first sector, and then extend it to the remaining five sectors by rotational symmetry.
	
	Assume first that \(A\) lies in the first sector. Then
	\begin{equation}
		A=r+s\omega,
	\end{equation}
	where
	\begin{equation}
		r\geq 0,\qquad s\geq 0,\qquad r+s\leq t.
		\label{eq:first_sector_conditions}
	\end{equation}
	Define two nodes
	\begin{equation}
		U=(r-t)+t\omega
		\label{eq:U_first_sector}
	\end{equation}
	and
	\begin{equation}
		V=-t+(t-s)\omega.
		\label{eq:V_first_sector}
	\end{equation}
	Then
	\begin{align}
		U-V
		&=\bigl((r-t)+t\omega\bigr)-\bigl(-t+(t-s)\omega\bigr)\\
		&=(r-t+t)+\bigl(t-(t-s)\bigr)\omega\\
		&=r+s\omega\\
		&=A.
	\end{align}
	Thus, \(A\) is the difference of \(U\) and \(V\). It remains to show that both \(U\) and \(V\) are boundary nodes.
	
	For \(U\), the coordinate pair is
	\begin{equation}
		(r-t,\ t).
	\end{equation}
	Using \eqref{eq:ej_hex_distance}, we obtain
	\begin{align}
		D(U)
		&=\max\{|r-t|,\ |t|,\ |(r-t)+t|\}\\
		&=\max\{t-r,\ t,\ r\}\\
		&=t,
	\end{align}
	because \(0\leq r\leq t\). Hence,
	\begin{equation}
		U\in B_t.
	\end{equation}
	
	For \(V\), the coordinate pair is
	\begin{equation}
		(-t,\ t-s).
	\end{equation}
	Again using \eqref{eq:ej_hex_distance}, we obtain
	\begin{align}
		D(V)
		&=\max\{|-t|,\ |t-s|,\ |-t+(t-s)|\}\\
		&=\max\{t,\ t-s,\ s\}\\
		&=t,
	\end{align}
	because \(0\leq s\leq t\). Hence,
	\begin{equation}
		V\in B_t.
	\end{equation}
	
	Therefore, every node \(A\) in the first sector can be expressed as
	\begin{equation}
		A=U-V
	\end{equation}
	for some \(U,V\in B_t\).
	
	Now consider a node \(A\) in any other sector. The remaining sectors are obtained from the first sector by multiplication by one of the six EJ units
	\begin{equation}
		1,\ \omega,\ \omega^2,\ -1,\ -\omega,\ -\omega^2.
	\end{equation}
	Multiplication by any EJ unit is a graph automorphism of \(H_n\). It preserves adjacency, graph distance, and the distance-\(t\) boundary. Therefore, if a rotated version of \(A\), denoted by \(A'\), lies in the first sector and has a decomposition
	\begin{equation}
		A'=U'-V'
	\end{equation}
	with
	\begin{equation}
		U',V'\in B_t,
	\end{equation}
	then rotating back gives
	\begin{equation}
		A=U-V
	\end{equation}
	with
	\begin{equation}
		U,V\in B_t.
	\end{equation}
	
	Thus, every node of \(H_n\) can be represented as the difference of two graph-distance-\(t\) boundary nodes. Hence,
	\begin{equation}
		B_t-B_t=H_n.
	\end{equation}
\end{proof}

\subsection{Two-Fault Re-Rooting Theorem}

By Lemma~\ref{lem:boundary_difference}, every relative displacement in \(H_n\) can be written as the difference of two boundary nodes.

\begin{theorem}[Existence of a Common Distance-$t$ New Source]
\label{thm:twofault}
Let $F_1$ and $F_2$ be any two faulty nodes in the dense EJ network $H_n$ with diameter $t=n-1$. Then there exists a node $NS\in H_n$ such that
\begin{equation}
d(F_1,NS)=t
\end{equation}
and
\begin{equation}
d(F_2,NS)=t.
\end{equation}
\end{theorem}

\begin{proof}
Let
\begin{equation}
A=F_2-F_1 \pmod{\alpha}.
\end{equation}
By Lemma~\ref{lem:boundary_difference}, there exist $U,V\in B_t$ such that
\begin{equation}
A=U-V \pmod{\alpha}.
\end{equation}
Choose
\begin{equation}
NS=F_1+U \pmod{\alpha}.
\end{equation}
Since $U\in B_t$, translation invariance gives
\begin{equation}
d(F_1,NS)=d(0,U)=t.
\end{equation}
Also,
\begin{align}
NS-F_2
&=F_1+U-F_2 \pmod{\alpha}\\
&=U-(F_2-F_1) \pmod{\alpha}\\
&=U-A \pmod{\alpha}\\
&=V \pmod{\alpha}.
\end{align}
Since $V\in B_t$, we obtain
\begin{equation}
d(F_2,NS)=d(0,V)=t.
\end{equation}
Therefore, $NS$ is a common graph-distance-$t$ node for $F_1$ and $F_2$.
\end{proof}

\begin{corollary}
For any one-node or two-node static fault configuration in $H_n$, the proposed re-rooting method can select a new source node such that all faulty nodes are leaf-level nodes in the subsequent one-to-all broadcast.
\end{corollary}

\subsection{Three-Fault Limitation}
The previous theorem gives a deterministic guarantee for any pair of faulty nodes. However, the same guarantee does not hold for arbitrary triples of faulty nodes.

\begin{proposition}[Three-Fault Counterexample]
\label{prop:threefault}
The two-fault re-rooting guarantee does not generally extend to arbitrary three-node fault configurations in dense EJ networks.
\end{proposition}

\begin{proof}
Consider $H_4$, for which
\begin{equation}
t=n-1=3
\end{equation}
and
\begin{equation}
N=3n^2-3n+1=37.
\end{equation}
Using the circulant representation with jumps
\begin{equation}
3,\quad 4,\quad 7,
\end{equation}
the graph-distance-$3$ boundary of node $0$ is
\begin{align}
B_3(0)=\{&2,5,9,12,13,15,16,17,18,19,20,21,\nonumber\\
&22,24,25,28,32,35\}.
\end{align}
Now choose three faulty nodes with integer labels
\begin{equation}
F_1=0,\quad F_2=5,\quad F_3=14.
\end{equation}
By vertex transitivity, the graph-distance-$3$ boundary of node $F$ is obtained by shifting $B_3(0)$ by $F$ modulo $37$. Thus,
\begin{align}
B_3(F_1)=\{&2,5,9,12,13,15,16,17,18,19,20,21,\nonumber\\
&22,24,25,28,32,35\},
\end{align}
\begin{align}
B_3(F_2)=\{&0,3,7,10,14,17,18,20,21,22,23,24,\nonumber\\
&25,26,27,29,30,33\},
\end{align}
and
\begin{align}
B_3(F_3)=\{&1,2,5,9,12,16,19,23,26,27,29,30,\nonumber\\
&31,32,33,34,35,36\}.
\end{align}
A valid re-rooted source for all three faulty nodes would need to lie in
\begin{equation}
B_3(F_1)\cap B_3(F_2)\cap B_3(F_3).
\end{equation}
However, direct intersection gives
\begin{equation}
B_3(F_1)\cap B_3(F_2)\cap B_3(F_3)=\emptyset.
\end{equation}
Therefore, no node $NS$ satisfies
\begin{equation}
d(NS,F_1)=d(NS,F_2)=d(NS,F_3)=3.
\end{equation}
This proves that the deterministic guarantee for two faulty nodes cannot be extended to arbitrary triples.
\end{proof}

\section{Complexity and Communication Cost}
\label{sec:complexity}

The proposed method consists of two phases: source selection and message dissemination.

\subsection{Source-Selection Cost}
For one faulty node, the algorithm scans at most $6t$ boundary candidates. Therefore,
\begin{equation}
T_1(t)=O(t).
\end{equation}
For two faulty nodes, the algorithm also scans at most $6t$ candidates and tests whether the translated candidate is on the boundary. Using the distance formula in \eqref{eq:hex_distance_max}, each membership test is constant time. Hence,
\begin{equation}
T_2(t)=O(t).
\end{equation}
Since \(N=3t^2+3t+1\), the source-selection cost is linear in the diameter and sublinear in the network size; equivalently, \(O(t)=O(\sqrt{N})\).

\subsection{Communication Cost}
The standard one-to-all broadcast in $H_n$ completes in
\begin{equation}
B_{\mathrm{baseline}}(t)=t
\end{equation}
parallel communication steps.

The proposed method first routes the message from the original source $S$ to the selected new source $NS$. Since the network diameter is $t$, the relocation distance satisfies
\begin{equation}
R(t)=d(S,NS)\leq t.
\end{equation}
After the message reaches $NS$, the standard one-to-all broadcast is executed from $NS$ in $t$ steps. Therefore, the total worst-case communication time is
\begin{align}
B_{\mathrm{total}}(t)
&=R(t)+B_{\mathrm{baseline}}(t)\\
&\leq t+t\\
&=2t.
\end{align}
Thus, the proposed method provides deterministic one- and two-node fault tolerance while increasing the worst-case broadcast time by at most one additional diameter.

\subsection{Why the Relocation Path Is Safe}
If a faulty node $F$ is placed at distance $t$ from $NS$, then it cannot be an internal node on a shortest path from the original source $S$ to $NS$. If it were an internal node on such a path, then the remaining distance from $F$ to $NS$ along that path would be strictly less than $t$. This would contradict the construction condition $d(F,NS)=t$. Therefore, the one-to-one relocation phase is not blocked by the faulty nodes selected by the re-rooting condition.

\section{Experimental Evaluation}
\label{sec:evaluation}

In this section, we evaluate the proposed re-rooting-based fault-tolerant broadcasting method in dense Eisenstein--Jacobi networks. The goal of the evaluation is to compare the reliability and overhead of the proposed method with the original one-to-all broadcasting algorithm under static node-failure scenarios.

\subsection{Experimental Setup}
\label{subsec:experimental_setup}

We conducted simulations on dense Eisenstein--Jacobi networks generated by
\begin{equation}
	\alpha=n+(n-1)\omega.
\end{equation}
The tested diameter values were
\begin{equation}
	t=10,25,50,100,200,
\end{equation}
corresponding to
\begin{equation}
	n=11,26,51,101,201.
\end{equation}
Thus, the tested network sizes were
\begin{equation}
	N=331,1951,7651,30301,120601.
\end{equation}

For each network size, two node-failure scenarios were considered:
\begin{itemize}
	\item Single-node failure: one faulty node is present in the network.
	\item Two-node failures: two distinct faulty nodes are present in the network.
\end{itemize}

For each value of \(t\) and each fault scenario, four fault-placement modes were tested:
\begin{itemize}
	\item Random: faulty nodes are selected uniformly at random.
	\item Near-source: faulty nodes are selected close to the original source.
	\item Critical-position: faulty nodes are selected from positions likely to interrupt the broadcast forwarding process.
	\item Close-pair: faulty nodes are selected from locally clustered regions.
\end{itemize}

For each exact configuration of diameter, fault count, and fault-placement mode, 1000 independent trials were performed. The full experiment contains
\begin{equation}
	5\times 2\times 4\times 1000=40000
\end{equation}
trials.

\subsection{Simulation Procedure}

For each tested value of \(n\), the simulator first constructs the dense EJ network
\begin{equation}
	H_n=\mathbb{Z}[\omega]/(n+(n-1)\omega).
\end{equation}
Each node is represented by an EJ coordinate \(x+y\omega\), and adjacency is defined using the six EJ directions
\begin{equation}
	\pm 1,\ \pm\omega,\ \pm\omega^2.
\end{equation}

For each trial, the faulty node set is generated according to the selected fault-placement mode. The original source node is fixed as node 0, and the same faulty set is used to evaluate both the baseline and proposed methods.

A trial is considered successful if every non-faulty node receives the broadcast message. Formally, if \(F\) is the set of faulty nodes and \(R\) is the set of nodes reached by the broadcast, then the trial is successful when
\begin{equation}
	H_n\setminus F \subseteq R.
\end{equation}

\subsection{Boundary-Difference Validation}

Before evaluating broadcast performance, we validate the main theoretical condition used by the proposed method. For each tested network size, the simulator constructs the distance-\(t\) boundary and verifies that $B_t-B_t=H_n$.

\begin{table}[H]
	\caption{Validation of the boundary-difference property in dense EJ networks.}
	\label{tab:boundary_validation}
	\centering
	\begin{tabular}{c c c c c}
		\hline
		\(n\) & \(t\) & \(N\) & \(|B_t|\) & \(|B_t-B_t|\) \\
		\hline
		11  & 10  & 331    & 60   & 331 \\
		26  & 25  & 1951   & 150  & 1951 \\
		51  & 50  & 7651   & 300  & 7651 \\
		101 & 100 & 30301  & 600  & 30301 \\
		201 & 200 & 120601 & 1200 & 120601 \\
		\hline
	\end{tabular}
\end{table}

The results in Table~\ref{tab:boundary_validation} support the theoretical result that the distance-\(t\) boundary difference set covers the entire dense EJ network.

\subsection{Broadcast Success and Reachability}

Table~\ref{tab:success_reach} summarizes the broadcast success rate and average reachability of the baseline and proposed methods across all tested fault-placement modes.

\begin{table}[H]
	\caption{Broadcast performance under one- and two-node failures across all fault-placement modes.}
	\label{tab:success_reach}
	\centering
	\begin{adjustbox}{max width=\textwidth}
	\begin{tabular}{c c c c c c c c c}
		\hline
		\(n\) & \(t\) & Faults & \(N\) & Trials & Baseline (\%) & Proposed (\%) & Avg. Base Reach & Avg. Prop. Reach \\
		\hline
		11  & 10  & 1 & 331    & 4000 & 9.125 & 100.0 & 318.952    & 330 \\
		26  & 25  & 1 & 1951   & 4000 & 3.875 & 100.0 & 1920.667   & 1950 \\
		51  & 50  & 1 & 7651   & 4000 & 2.300 & 100.0 & 7589.553   & 7650 \\
		101 & 100 & 1 & 30301  & 4000 & 1.025 & 100.0 & 30188.283  & 30300 \\
		201 & 200 & 1 & 120601 & 4000 & 0.500 & 100.0 & 120378.393 & 120600 \\
		\hline
		11  & 10  & 2 & 331    & 4000 & 3.925 & 100.0 & 307.554    & 329 \\
		26  & 25  & 2 & 1951   & 4000 & 0.975 & 100.0 & 1891.973   & 1949 \\
		51  & 50  & 2 & 7651   & 4000 & 0.650 & 100.0 & 7528.753   & 7649 \\
		101 & 100 & 2 & 30301  & 4000 & 0.300 & 100.0 & 30075.920  & 30299 \\
		201 & 200 & 2 & 120601 & 4000 & 0.225 & 100.0 & 120198.111 & 120599 \\
		\hline
	\end{tabular}
	\end{adjustbox}
\end{table}

The proposed method achieves complete delivery in all one- and two-node failure configurations because the selected new source places every faulty node on the distance-\(t\) boundary.

\begin{figure}[H]
	\centering
	\includegraphics[width=0.6\linewidth]{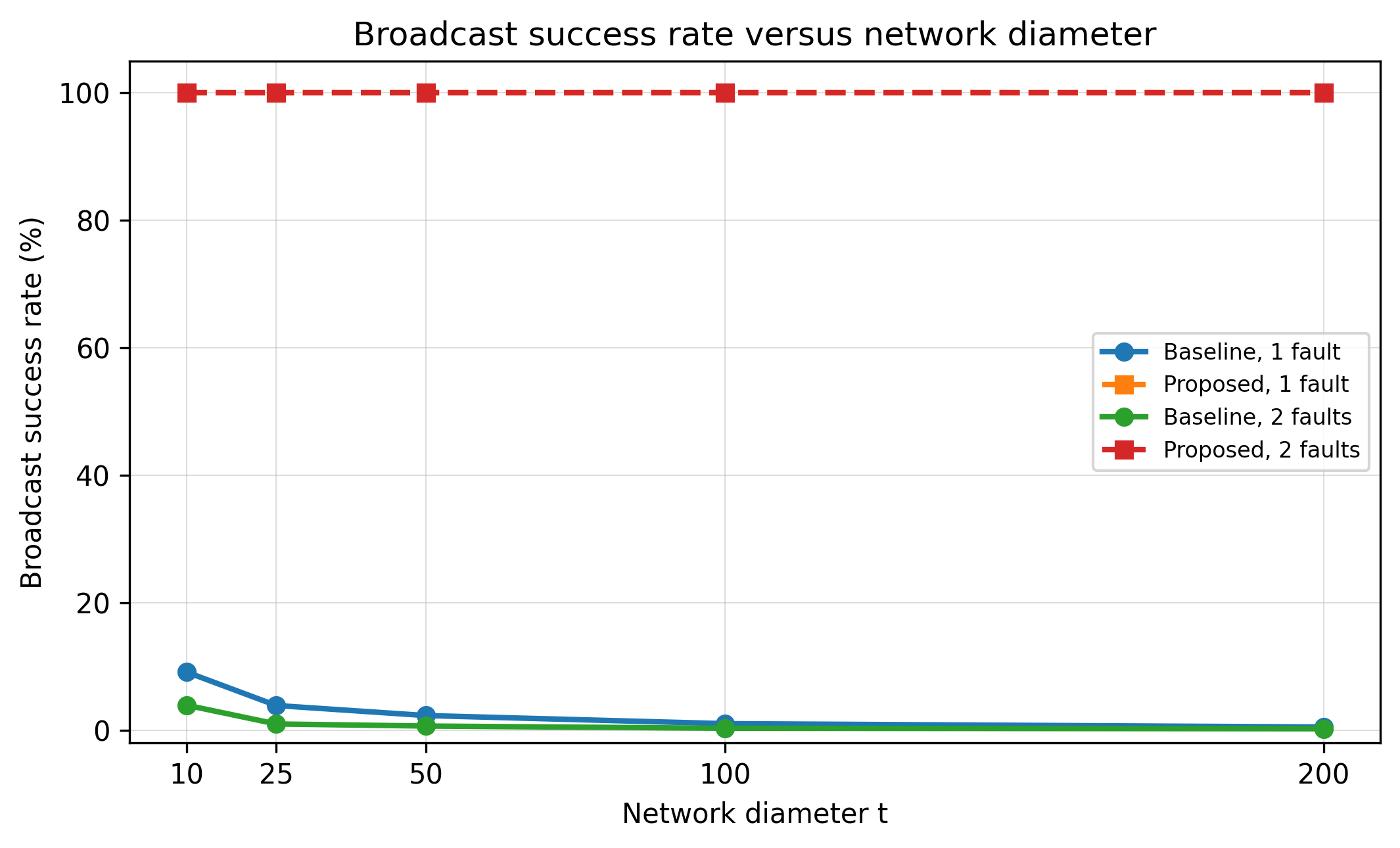}
	\caption{Broadcast success rate versus network diameter \(t\). The proposed re-rooting method achieves complete delivery in all tested one- and two-node failure scenarios, while the baseline broadcast success rate decreases as the network size increases.}
	\label{fig:success_rate}
\end{figure}

Figure~\ref{fig:success_rate} shows the broadcast success rate as the network size increases.

\subsection{Reachability Variation}

Table~\ref{tab:std_reach} reports the standard deviation of the number of reached nodes across trials.

\begin{table}[H]
	\caption{Reachability variation under fault-placement experiments.}
	\label{tab:std_reach}
	\centering
	\begin{tabular}{c c c c c}
		\hline
		\(n\) & \(t\) & Faults & Base SD & Prop. SD \\
		\hline
		11  & 10  & 1 & 15.777   & 0.000 \\
		26  & 25  & 1 & 64.026   & 0.000 \\
		51  & 50  & 1 & 183.218  & 0.000 \\
		101 & 100 & 1 & 482.064  & 0.000 \\
		201 & 200 & 1 & 1386.801 & 0.000 \\
		11  & 10  & 2 & 24.487   & 0.000 \\
		26  & 25  & 2 & 93.407   & 0.000 \\
		51  & 50  & 2 & 270.004  & 0.000 \\
		101 & 100 & 2 & 694.483  & 0.000 \\
		201 & 200 & 2 & 1806.523 & 0.000 \\
		\hline
	\end{tabular}
\end{table}

The proposed method has zero reachability variation when it reaches exactly \(N-|F|\) non-faulty nodes in every successful trial. In contrast, the baseline method may show larger variation because its performance depends strongly on the locations of the faulty nodes relative to the original broadcast tree.

\subsection{Communication and Computational Overhead}

Table~\ref{tab:overhead} reports the communication-step and computational overhead of the proposed method.

\begin{table}[H]
	\caption{Communication-step and computational overhead of the proposed re-rooting method. The checked-candidate count is bounded by \(6t\), confirming the \(O(t)\) source-selection behavior.}
	\label{tab:overhead}
	\centering
	\begin{adjustbox}{max width=\textwidth}
	\begin{tabular}{c c c c c c c c c}
		\hline
		\(n\) & \(t\) & Faults & \(N\) & Avg. Reloc. Hops & Avg. Total Steps & Min--Max Steps & Avg. Runtime (ms) & Avg. Checked \\
		\hline
		11  & 10  & 1 & 331    & 7.633   & 17.633  & 10--20   & 0.002129 & 1.000 \\
		26  & 25  & 1 & 1951   & 18.369  & 43.369  & 26--50   & 0.002267 & 1.000 \\
		51  & 50  & 1 & 7651   & 36.735  & 86.735  & 51--100  & 0.002395 & 1.000 \\
		101 & 100 & 1 & 30301  & 72.856  & 172.856 & 103--200 & 0.003168 & 1.000 \\
		201 & 200 & 1 & 120601 & 145.151 & 345.151 & 204--400 & 0.005129 & 1.000 \\
		\hline
		11  & 10  & 2 & 331    & 7.649   & 17.649  & 10--20   & 0.003954 & 7.278 \\
		26  & 25  & 2 & 1951   & 18.656  & 43.656  & 25--50   & 0.005808 & 19.065 \\
		51  & 50  & 2 & 7651   & 36.507  & 86.507  & 51--100  & 0.009112 & 39.500 \\
		101 & 100 & 2 & 30301  & 72.978  & 172.978 & 101--200 & 0.019914 & 78.585 \\
		201 & 200 & 2 & 120601 & 146.006 & 346.006 & 201--400 & 0.036946 & 151.961 \\
		\hline
	\end{tabular}
	\end{adjustbox}
\end{table}

The fault-free EJ broadcast completes in \(t=n-1\) parallel steps. The proposed method adds one preliminary relocation phase from the original source \(S\) to the selected new source \(NS\). Since the diameter of \(H_n\) is \(t\), the relocation distance satisfies
\begin{equation}
	d(S,NS)\leq t.
\end{equation}
Therefore, the total broadcast time satisfies
\begin{equation}
	B_{\mathrm{total}}(t)\leq 2t.
\end{equation}

The maximum total-step values in Table~\ref{tab:overhead} confirm this bound. The computational overhead comes from selecting \(NS\). Since the proposed source-selection algorithm scans a linear-size boundary of size \(6t\), the source-selection cost is \(O(t)\).

\begin{figure}[H]
	\centering
	\begin{subfigure}{0.48\linewidth}
		\centering
		\includegraphics[width=\linewidth]{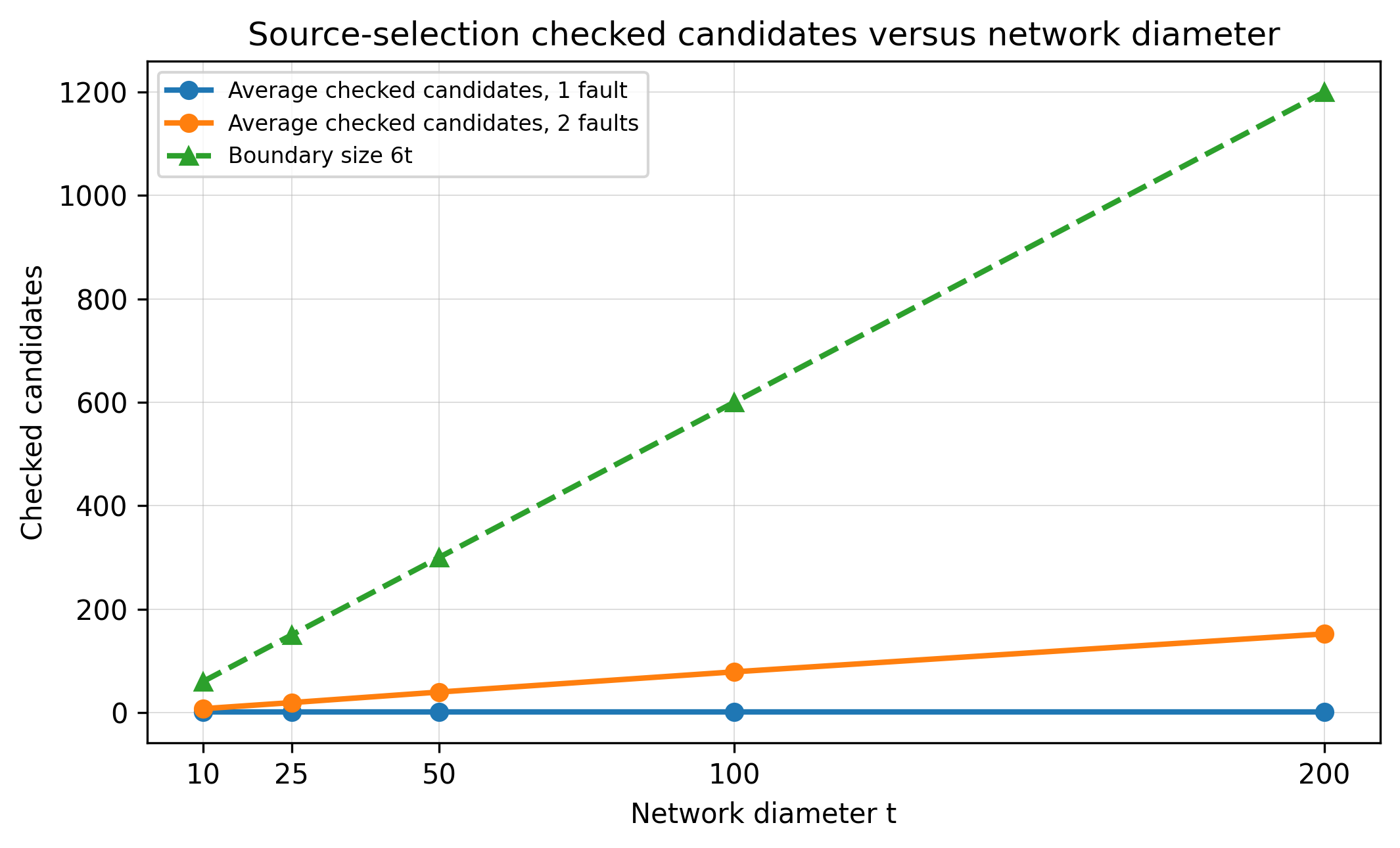}
		\caption{Checked candidates}
	\end{subfigure}
	\hfill
	\begin{subfigure}{0.48\linewidth}
		\centering
		\includegraphics[width=\linewidth]{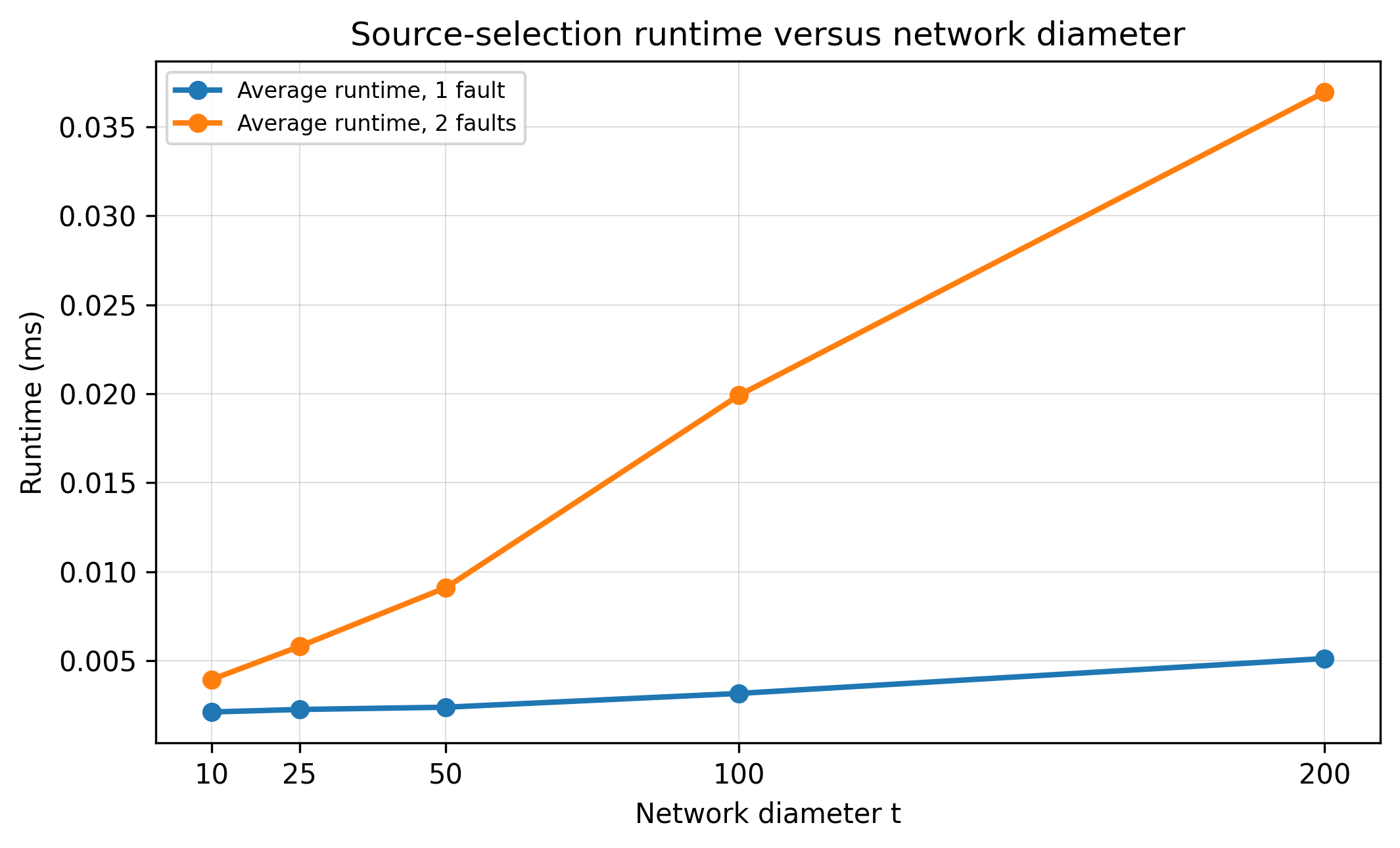}
		\caption{Runtime}
	\end{subfigure}
	\caption{Source-selection overhead versus network diameter \(t\). (a) The number of checked candidates remains bounded by the graph-distance-\(t\) boundary size \(6t\), confirming the \(O(t)\) source-selection behavior. (b) The measured runtime remains small across all tested network sizes.}
	\label{fig:source_selection_overhead}
\end{figure}

Figure~\ref{fig:source_selection_overhead} (a) shows that the number of checked candidates is bounded by \(6t\), while Figure~\ref{fig:source_selection_overhead} (b) shows the measured source-selection runtime.

\subsection{Comparative Discussion}
\label{subsec:comparative_discussion}

Although the experimental comparison in this work focuses on the original non-fault-tolerant one-to-all broadcasting algorithm, it is useful to position the proposed re-rooting method relative to common fault-tolerant communication approaches.

Table~\ref{tab:qualitative_comparison} summarizes the qualitative differences between these approaches and the proposed re-rooting method.

\begin{table}[H]
	\caption{Qualitative comparison with common fault-tolerant broadcasting strategies.}
	\label{tab:qualitative_comparison}
	\centering
	\begin{tabular}{l c c c}
		\hline
		Method & Extra & Runtime & Changed \\
		& structures & adaptation & broadcast \\
		\hline
		Redundant trees & Yes & Low/Med. & Yes \\
		Multiple paths & Yes & Medium & Yes \\
		Adaptive rerouting & No/Partial & High & Yes \\
		Local recovery & Partial & Medium & Yes \\
		Proposed re-rooting & No & Low & No \\
		\hline
	\end{tabular}
\end{table}

The main advantage of the proposed method is its simplicity. It is designed specifically for dense EJ networks, whose algebraic and geometric structure allows the source to be relocated so that one or two faulty nodes lie on the graph-distance-\(t\) boundary. Since boundary nodes are leaves in the one-to-all broadcast tree, they do not need to forward the message. Therefore, the original broadcast procedure can be preserved.

\subsection{Failure Model and Limitations}

The proposed method is designed for static node-failure scenarios. In this model, the set of faulty nodes is known before the broadcast begins and does not change during message propagation.

This work focuses on one- and two-node failures. For any pair of faulty nodes in a dense EJ network, the boundary-difference theorem guarantees the existence of a node \(NS\) whose graph distance from both faulty nodes is exactly \(t\). Therefore, the proposed method provides a deterministic guarantee for all one- and two-node fault configurations.

The present model does not directly address link failures. A failed link may interrupt communication even when both endpoint nodes remain operational. Extending the re-rooting idea to link failures would require a different condition, possibly based on avoiding failed edges along the six directional broadcast sectors rather than selecting a source that makes faulty nodes leaves.

The model also assumes that failures are static during the broadcast. If a new node fails after the message has already been relocated to \(NS\), the selected source may no longer satisfy the required distance condition. Handling dynamic or transient failures would require repeated re-rooting, online failure detection, or adaptive recovery during message propagation. These extensions are outside the scope of the present work.

\subsection{Practical Implications}

The proposed method is particularly suitable for systems in which the underlying communication topology is regular and the broadcast algorithm is already optimized for parallel propagation. Dense EJ networks have degree six, small diameter, regular structure, and efficient one-to-all communication. In such systems, adding redundant spanning trees or maintaining several backup routes may increase routing-table storage, control overhead, and implementation complexity.

By contrast, the proposed method preserves the original EJ one-to-all broadcast algorithm and adds only a lightweight source-selection and relocation phase. This property is attractive for NoC-based chip multiprocessors, many-core accelerators, distributed edge-computing systems, and embedded parallel systems where predictable communication behavior and low implementation overhead are important.

\section{Conclusion}
\label{sec:conclusion}

This paper presented a re-rooting-based fault-tolerant broadcasting method for dense Eisenstein--Jacobi networks. The method relocates the effective broadcast source so that faulty nodes are placed at graph distance equal to the network diameter from the new source. Under the standard one-to-all broadcasting algorithm, these nodes become leaf-level nodes and are not required to forward packets.

For the single-fault case, a valid new source can be selected directly from the graph-distance-$t$ boundary of the faulty node. For the two-fault case, we proved that for any pair of faulty nodes in a dense EJ network, there exists a node whose graph distance from both faulty nodes is exactly $t$. The proof is based on a boundary-difference coverage property showing that every network node can be written as the difference of two distance-$t$ boundary nodes.

The proposed source-selection algorithm runs in linear time with respect to the network diameter \(t\). The communication cost consists of one relocation phase of at most \(t\) steps followed by the standard \(t\)-step broadcast. Hence, the total worst-case broadcast time is bounded by \(2t\).

The paper also identified the theoretical limitation of the approach. An explicit counterexample in $H_4$ shows that a common graph-distance-$t$ source does not necessarily exist for arbitrary three-node fault configurations. Therefore, the proposed method is positioned as a deterministic lightweight solution for static one- and two-node failures in dense degree-six EJ networks.

Future work will investigate whether additional structure in the dense EJ boundary can be used to further reduce the source-selection cost for special fault configurations. Additional directions include extensions to link failures, dynamic failures, and partial three-fault classes that admit common boundary re-rooting, as well as implementation-level comparisons with specific redundant-tree and adaptive-routing methods under a unified EJ network simulation framework.

\section*{Acknowledgment}
The authors would like to acknowledge the support of Kuwait University and its Department of Computer Science.

\end{document}